\documentclass[envcountsame]{llncs}

\usepackage[utf8]{inputenc}
\usepackage{amsmath}
\usepackage{amssymb}
\usepackage{rotating}
\usepackage{amsfonts}
\usepackage{hyperref}
\usepackage{caption}
\usepackage{subcaption}
\usepackage{stmaryrd}
\usepackage{tikz}
\usepackage{enumerate}
\newcommand{\Defeq}{\stackrel{\mathrm{def}}{=}}

\newcommand{\dtrans}[2]{\lower.2em\hbox{$\xrightarrow{#1/#2}$}}
\newcommand{\dTrans}[2]{\lower.2em\hbox{$\xRightarrow{#1/#2}$}}
\newcommand{\epstrans}[2]{\lower.2em\hbox{$\xrightarrow{\epsilon_{#1,#2}}$}}

\newcommand{\pre}[1]{{^\circ{#1}}}
\newcommand{\post}[1]{{#1^\circ}}

\newcommand{\source}[1]{{^\bullet{#1}}}
\newcommand{\target}[1]{{#1^\bullet}}

\newcommand{\marking}[2]{{[#1]}_{#2}}

\newcommand{\semanticsOf}[1]{\llbracket{#1}\rrbracket}

\newcommand{\figref}[1]{Fig.~\ref{#1}}
\newcommand{\exmref}[1]{Example~\ref{#1}}
\newcommand{\secref}[1]{Sec.~\ref{#1}}

\makeatletter
\def\moverlay{\mathpalette\mov@rlay}
\def\mov@rlay#1#2{\leavevmode\vtop{%
\baselineskip\z@skip \lineskiplimit-\maxdimen
\ialign{\hfil$#1##$\hfil\cr#2\crcr}}}
\makeatother

\newcommand{\ten}{\otimes}
\newcommand{\union}{\mathrel{\cup}}
\newcommand{\intersection}{\mathrel{\cap}}

\newcommand{\diag}{\ensuremath{\mathsf{\Delta}}}

\newcommand{\comp}{\mathrel{;}}
\newcommand{\ldiag}{\ensuremath{\mathsf{\Lambda}}}

\newcommand{\rightEnd}{\ensuremath{\pmb{\bot}}}

\newcommand{\lzero}{\ensuremath{\,\pmb{\uparrow}\,}}
\newcommand{\rzero}{\ensuremath{\,\pmb{\downarrow}\,}}

\newcommand{\rring}[1]{\ensuremath{\mathbb{#1}}}
\newcommand{\N}{\rring{N}}

\newcommand{\setof}[1]{\left\lbrace\,#1\,\right\rbrace}

\newcommand{\bnfSep}{\ |\ }
\newcommand{\bnfEq}{\ ::=\ }

\newcommand{\ordinal}[1]{{[#1]}}

\newcommand{\places}[1]{\mathsf{places}(#1)}
\newcommand{\trans}[1]{\mathsf{trans}(#1)}
\newcommand{\contention}{\ensuremath{\bowtie}}
\newcommand{\fat}[1]{\mathbf{#1}}

\newcommand{\wiringVarEnv}{\mathcal{V}}
\newcommand{\wiringTree}{\ensuremath{t}}
\newcommand{\wiringDecomp}{\ensuremath{\left(\wiringTree,\,\wiringVarEnv\right)}}

\newcommand{\portset}[1]{\left\langle\,#1\,\right\rangle}
\newcommand{\footprint}[1]{\left[\,#1\,\right]}

\newcommand{\outport}[1]{{#1_\blacktriangleright}}
\newcommand{\inport}[1]{{#1_\blacktriangleleft}}
\newcommand{\leftb}[1]{{#1_{\mathsf{L}}}}
\newcommand{\rightb}[1]{{#1_{\mathsf{R}}}}

\newcommand{\minsync}[2]{\mathsf{Synch}(#1,#2)}

\newcommand{\ports}[1]{\mathsf{ports}(#1)}
\newcommand{\conn}[1]{\mathsf{conn}(#1)}
\newcommand{\connwithnet}[2]{\mathsf{conn}^{#1}(#2)}
\newcommand{\eports}[1]{\mathsf{eports}(#1)}
\newcommand{\connr}[2]{\mathsf{conn}_{#1}(#2)}
\newcommand{\bconn}[2]{\mathsf{bconn}_{#1}(#2)}
\newcommand{\network}[2]{\mathsf{network}_{#1}(#2)}
\newcommand{\dimension}[1]{\mathsf{dim}(#1)}

\title{Decomposing Petri nets}
\author{Julian Rathke \and Pawe{\l} Soboci{\'n}ski \and Owen Stephens}
\institute{ECS, University of Southampton, UK}
\begin{document}
\maketitle
\begin{abstract}
In recent work, the second and third authors introduced a technique for
reachability checking in 1-bounded Petri nets, based on \emph{wiring
decompositions}, which are expressions in a fragment of the compositional algebra of \emph{nets
with boundaries}. Here we extend the technique to the full
algebra and introduce the related structural property of \emph{decomposition
width} on directed hypergraphs. Small decomposition width is necessary for the
applicability of the reachability checking algorithm. We give examples
of families of nets with constant decomposition width and develop the underlying
theory of decompositions.
\end{abstract}

\section*{Introduction}

Model checking asynchronous systems is notoriously susceptible to state
explosion. Historically, Petri nets are one of the most popular formalisms for modelling asynchronous systems. Several model checking problems
reduce to checking reachability in (bounded) Petri nets, where state explosion manifests itself in the fact that the set of markings is exponential in the number of places.
Our approach to the problem of state explosion is to check reachability of a
net in a divide-and-conquer, dynamic programming style~\cite{Sobocinski2013}
by considering decompositions of the net into smaller subnets and
checking reachability locally. Clearly, this approach relies heavily on a
principled notion of Petri net decomposition, which is the topic of this
paper.

In~\cite{Soboci'nski2010} the second author introduced a compositional algebra
of 1-bounded Petri nets, called \emph{nets with boundaries}, which was later
extended by Bruni, Melgratti and Montanari~\cite{Bruni2011} to cover P/T nets;
see~\cite{Bruni2012} for a complete exposition. A net with boundaries induces a
labelled transition system (LTS) where the states correspond to the
markings of the net and the transitions witness the firings of
independent sets of net transitions. Following the process calculus tradition,
the labels of LTS transitions describe synchronisations with the environment.

In recent work~\cite{Sobocinski2013}, the second and third author used this
algebra to check reachability for 1-bounded nets. A decomposition of a net into
an expression in the algebra of nets with boundaries is called a \emph{wiring
decomposition}---concretely, it is a tree, with internal nodes labelled by the
two operations `$\comp$' and `$\otimes$' for composing nets with boundaries,
and the leaves labelled with individual nets with boundaries. For the purposes
of reachability, given a wiring decomposition, each component net's LTS is
considered as a non-deterministic finite automaton (NFA) with initial state the
initial (local) marking and final state the desired (local) marking. Because
the algebra is compositional, the NFA for the entire net can be obtained by
composing the NFAs of the individual component nets, following the structure of
the wiring decomposition. This underlying algebra of NFAs (transition systems)
is that of Span(Graph)~\cite{Katis1997a}.

If, given a net, a ``good'' wiring decomposition can be found
then characterising communication between components will require small
(w.r.t. the global statespace) amounts of information.
Once reachability is checked locally, local statespace can be discarded and thus state-explosion circumvented. Exposing the regular structure of a net, moreover, allows repeated work to be avoided: memoisation of local reachability checks on small component nets leads to better performance. As a result, in some
examples (see~\cite{Sobocinski2013}) reachability checking is \emph{linear} in
the size of the net, even when \emph{the length of the minimal firing sequence
required to reach the desired marking is non-linear}. The approach can thus
sometimes outperform classical techniques for checking reachabilty, for
instance, those based on the unfolding technique, originally pioneered by
McMillan~\cite{McMillan1995}.

The applicability of the technique described in~\cite{Sobocinski2013} is thus
closely related to the problem of obtaining wiring decompositions of nets. When
translating a net with boundaries to an LTS, its size depends on two
factors: \textit{(i)}{ the number of places and \textit{(ii)} the
size of its boundaries. The size of the LTS statespace is typically
 exponential in the number of places, as states correspond to markings.
The size of the set of LTS labels is exponential in the size of the
boundary.

What is a ``good'' wiring decomposition? Recall that a wiring decomposition is
a tree.  Firstly, the leaves of this tree are subnets and, in order
to keep the size of the LTSs manageable, each leaf should have few places, and
a small boundary. Secondly, each subtree of the wiring decomposition should
result in a net with a small boundary, to keep the size of the label set small
when checking the compositions of subnets. Thirdly, the minimised statespaces
of (NFAs of) subtrees should ``grow slowly'' towards the root, so that state
explosion is avoided.

The first two conditions amount to a \emph{structural
property}\footnote{Analogously to how pathwidth and treewidth are
structural properties of undirected graphs. Treewidth is well known in the CONCUR community through Courcelle's theorem~\cite{Courcelle1990}.} on the underlying net, considered as a directed hypergraph. We call this property \emph{decomposition width}: a
net (or directed hypergraph) has \emph{decomposition width} $k$ iff it has a
wiring decomposition of width $k$. The third condition is a \emph{semantic
property}: in particular, a net can have several decompositions of equal width
that perform differently with respect to the third criterion. Several examples
are given in~\cite{Sobocinski2013}.

In this paper, we concentrate on the structural property of decomposition
width. We make use of the full algebra of nets with
boundaries~\cite{Bruni2012}, which allows us to cover more examples than
in~\cite{Sobocinski2013} where we considered a restricted variant. We discover
that sparsely connected nets, ``tree-like'' nets, but also cliques and related
``densely'' connected nets are all examples of families of nets that admit
decompositions of small width. By this we mean that there is some $k$ such that
the entire family of nets (of arbitrary size) has decomposition width $k$. We
also give an example of a family of grid nets that we conjecture not to admit bounded decomposition. Decomposition width is thus different to parameters which have previously been considered on nets, such as treewidth of the flow graph~\cite{Praveen2011}; (like treewidth, grids seem problematic, but unlike treewidth, cliques are not.)

Concretely, the contributions of this paper are:
\begin{itemize}
\item The full algebra of nets with
    boundaries~\cite{Soboci'nski2010,Sobocinski2013} is used with
    the reachability technique of~\cite{Sobocinski2013}. We thus extend the
    applicability of the technique to examples such as clique nets.
\item The structural property of decomposition width on nets (or, more
    generally, on directed hypergraphs) is introduced.
\item The theory of wiring decompositions
    is developed, which allows us to give lower bounds on boundary sizes
    in certain decompositions.
\end{itemize}

\paragraph{Structure of the paper.}
In \S\ref{sec:nets_with_boundaries} we recall and generalise the definition of nets with
boundaries. In \S\ref{sec:width} we introduce the notion of \emph{decomposition
width}, and explain its central role in the performance of our technique, which we
briefly recap in \S\ref{sec:recap}. We discuss an extension to the previously
considered net algebra in \S\ref{sec:extended_algebra}, using the full algebra
of nets with boundaries in order to apply our technique to more cases. In
\S\ref{sec:principles} we introduce the principles of decomposition, and use
them to show lower bounds for the size of decompositions in certain nets.

\section{Preliminaries}
\label{sec:prelims}

For $n\in\N$, let $\ordinal{n}=\{0,1,\dots,n-1\}$. Write $2^X$ for the powerset
of $X$ and $X+Y$ for the set $\{\,(x,0)\;|\;x\in X\,\}\cup\{\,(y,1)\;|\;y\in
Y\,\}$.

\begin{definition}[1-bounded Petri net]
A net $N$ is $(P,T,\pre{-},\post{-})$ where
\begin{itemize}
\item[-] $P$ is the set of places, $T$ is the set of transitions
\item[-] $\pre{-},\,\post{-}:T\to 2^P$ give, respectively, the pre- and
    post-sets of each transition.
\end{itemize}
\end{definition}
We write $\places{N}$ and $\trans{N}$ for the place and transition sets,
respectively, of $N$. Our underlying semantics is a step firing semantics
where independent sets of transitions can be fired together; to minimise
redundancy, we give the definition in~\eqref{eq:transitionRelation} in the
more general setting of nets with boundaries.


\section{Nets with boundaries}
\label{sec:nets_with_boundaries}


A net with boundaries~\cite{Soboci'nski2010} is a Petri net together with two
ordered sets of \emph{boundary ports}, to which net transitions can connect.
Nets with boundaries inherit the algebra of monoidal categories for
composition. In this paper we expand upon the previous exposition of nets with
boundaries in~\cite{Soboci'nski2010,Sobocinski2013}, by lifting the restriction
of~\cite{Sobocinski2013} that at most one transition can connect to any one
place on a boundary.

\begin{definition}[Net with boundaries]
A \emph{net with boundaries} $N:k\to l$ is $(P, T, k, l, \!\pre{-},
\post{-},\!\source{-},\target{-},\contention)$ where:
\begin{itemize}
\item[-] $(P,T,\pre{-},\post{-})$ is a 1-bounded Petri net
\item[-] $k,l\in\N$ are, respectively, the left and the right boundaries
\item[-] $\source{-}:T\to 2^{\ordinal{k}}$ and
    $\target{-}:T\to 2^{\ordinal{l}}$ connect each transition to, respectively, the
    left and the right boundary
\item[-] $\contention$ is a contention relation (see Definition~\ref{defn:contention} below).
\end{itemize}
\end{definition}
Isomorphism, $(N:k\to l) \cong (M:k\to l)$, is defined in the obvious way as
bijections between place sets and transition sets that respect pre and post
sets, boundary connections and contention. 1-bounded Petri nets $N$ can be
considered as nets with boundaries $N:0\to 0$ (with the minimal contention
relation).

\begin{remark}\label{rmk:restriction}
In~\cite{Sobocinski2013} we assumed that for any $t\neq t'\in T$,
$\source{t}\cap\source{t'}=\varnothing$ and
$\target{t}\cap\target{t'}=\varnothing$; i.e. no two transitions connect to the
same boundary port. In \secref{sec:extended_algebra}, we show that certain
nets admit better decompositions without this restriction.
\end{remark}

In order to leave out the assumption, we must recall the notion of
\emph{contention} between transitions, first proposed in~\cite{Bruni2012}.
Transitions in contention cannot fire concurrently. In ordinary nets, two
transitions are in contention precisely when they compete for a resource, for
instance they consume or produce a token at the same place. In nets with
boundaries, connecting two transitions to the same boundary port is another
source of contention. Examples and the mathematical foundations
of contention are given in~\cite{Sobocinski2013a}. Roughly speaking, contention is ``remembered'' in compositions; this is needed in order to ensure that net composition is compatible with the composition of underlying transition systems.

\begin{definition}[Contention Relation]\label{defn:contention}
    For a net $N$, a reflexive, symmetric relation, $\contention$, on
    $\trans{N}$ is said to be a \emph{contention relation}, if for all
    $(t, u) \in \trans{N} \times \trans{N}$ where at least one of the
    following holds
    \[
        \mathit{(i)}\  \pre{t} \intersection \pre{u} \neq \varnothing \quad
        \mathit{(ii)}\  \post{t} \intersection \post{u} \neq \varnothing \quad
        \mathit{(iii)}\  \source{t} \intersection \source{u} \neq \varnothing \quad
        \mathit{(iv)}\  \target{t} \intersection \target{u} \neq \varnothing.
	\]
    then $t \contention u$.
\end{definition}

\begin{remark}[Graphical representation]
See~\figref{fig:Tnk} and~\figref{fig:Tlnk} for several simple examples of
nets with boundaries. The graphical representation we use is non-standard and
deserves an explanation:
Concretely, each place is drawn as ``directed,'' having an \textit{in} and
\textit{out} port. Transitions are undirected links that connect an arbitrary
set of boundaries and place ports.
The benefit of doing this is that links, which are connected together during
composition, do not need to be directionally compatible in order to compose
two nets.
Instead, the places contain the firing direction information, localising
the firing semantics to subcomponents.
The preset of a transition is simply the set of places to which the transition is connected
via the \textit{out} port (a triangle pointing out of a place), symmetrically,
its postset is the set of places to which the transition is connected via the
\textit{in} port (a triangle pointing into a place.)
In order to distinguish individual transitions and increase legibility,
transitions are drawn with a small perpendicular mark.
\end{remark}

A transition set $U$ is \emph{mutually independent} (MI) if $\forall
u, v \in U.\; u \contention v \Rightarrow u = v$. Contention can be
lifted to \emph{sets} of mutually independent transitions: $U
\contention V$ iff $\exists u \in U, v \in V.\; u \contention v$. Mutually
independent transitions can fire concurrently: each net with boundaries $N:k\to
l$ determines an LTS\footnote{Originally described in Katis et
al~\cite{Katis1997b}.}, $\semanticsOf{N}$, whose transitions witness the step
firing semantics of the underlying net. The labels are pairs of binary strings
of length $k$ and $l$, respectively.  The states are markings of $N$, denoted
by $\marking{N}{X}$, where $X\subseteq \places{N}$. The transition relation is
defined\footnote{We equate binary strings of length $k$ with subsets of $[k]$,
in the obvious way.}: \begin{multline}\label{eq:transitionRelation}
\marking{N}{X} \dtrans{\alpha}{\beta} \marking{N}{X'} \Leftrightarrow\exists\text{ MI }U\subseteq T, \pre{U}\subseteq X,\,\post{U}\cap X=\varnothing,\,\\ X'= (X\backslash \pre{U})\cup \post{U},\, \source{U}=\alpha,\, \target{U}=\beta.
\end{multline}

In order to compose nets with boundaries along a common boundary, we recall the notion of \emph{synchronisation}. For sets of transitions $U\subseteq T$ we abuse notation and write $\pre{U}=\bigcup_{u\in U}\pre{u}$, and similarly for $\post{U}$, $\source{U}$ and $\target{U}$.
\begin{definition}[Synchronisations]\label{defn:synch}
    A synchronisation between two nets with boundaries $M:l\to m$, $N:m\to n$ is a pair $(U,V)$, $U \subseteq \trans{M}$ and $V \subseteq \trans{N}$, of mutually independent sets of transitions, such that $\target{U} = \source{V}$.

    Synchronisations inherit an ordering from the subset ordering, pointwise:
    $(U, V) \subseteq (U', V') \Defeq U \subseteq U' \wedge V \subseteq
    V'$. The trivial synchronisation is $(\varnothing, \varnothing)$. A
    synchronisation $(U, V)$ is \emph{minimal} when it is not trivial, and for
    all $(U', V') \subseteq (U, V)$, then $(U', V')$ is trivial or equal to
    $(U, V)$. Contention can be lifted to minimal synchronisations: $(U, V) \contention (U', V') \Defeq U \contention U' \vee V \contention V'$.
\end{definition}
Given $M:l\to m$, $N:m\to n$, let $\minsync{M}{N}$ be the set of minimal synchronisations. We can now define the two ways of composing nets with boundaries.
\begin{definition}[Composition along common boundary]
The composition of nets $M:l\to m$ and $N:m\to n$, $M;N:l\to n$ has the following components:
\begin{itemize}
    \item[-] the set of places is $\places{M} + \places{N}$.
    \item[-] the set of transitions is $\minsync{M}{N}$, the set of minimal synchronisations.
    \item[-] $\forall (U, V) \in \minsync{M}{N}, \pre{(U, V)} \Defeq \pre{U}+
        \pre{V}$ and $\post{(U,V)} \Defeq \post{U}+\post{V}$.
    \item[-] $\forall (U, V) \in \minsync{M}{N}, \source{(U, V)} \Defeq \source{U}$ and $\target{(U,V)} \Defeq \target{V}$.
    \item[-] Contention on minimal synchronisations as described in Definition~\ref{defn:synch}.
\end{itemize}
\end{definition}
\begin{definition}[Tensor product]
The tensor product of nets $M:l\to m$ and $N:k\to n$,
$M\otimes N:l+k \to m+ n$ has the following components:
\begin{itemize}
    \item[-] the set of places is $\places{M}+\places{N}$.
    \item[-] the set of transitions is $\trans{M}+\trans{N}$.
    \item[-] the preset, postset, and boundary maps are defined in the obvious way.
    \item[-] transitions in $\trans{M}+\trans{N}$ are in contention exactly when they are in contention in either $M$ or $N$.
\end{itemize}
\end{definition}
Both `$\comp$'-composition and `$\ten$'-composition are associative up-to
isomorphism. In examples we will make use of a exponentiation notation: given $N:l\to l$, we write $N^k$ for the `$\comp$'-composition of $N$ with itself $k$-times: $N \comp N \comp \dots \comp N$.

There are several compositionality results reported
in~\cite{Soboci'nski2010,Bruni2012,Sobocinski2013} (e.g. Theorem 3.8 of \cite{Soboci'nski2010}); essentially the idea is that firings of a composed net (as LTS transitions $\dtrans{\alpha}{\beta}$) are in direct correspondence with firings ($\dtrans{\alpha}{\gamma}$ and
$\dtrans{\gamma}{\beta}$) of components.


\begin{figure}[h]
    \centering
    \begin{subfigure}{0.45\textwidth}
        \centering
        \includegraphics[height=3.5cm]{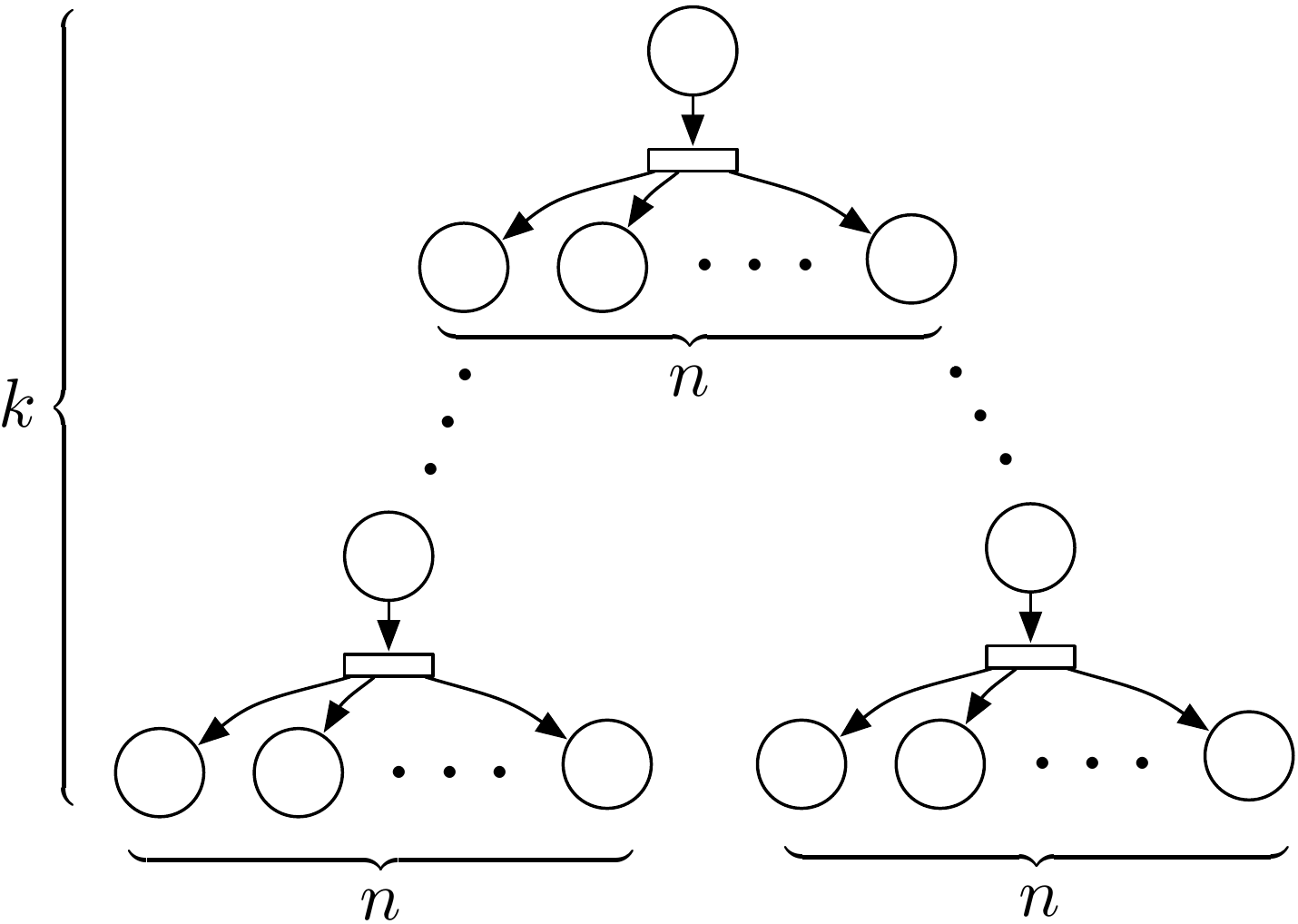}
        \caption{$T_\diag^{n,k}$ - single transitions between parent and
        children. \label{fig:tdiag}}
    \end{subfigure}
    \hfill
    \begin{subfigure}{0.45\textwidth}
        \centering
        \includegraphics[height=3.5cm]{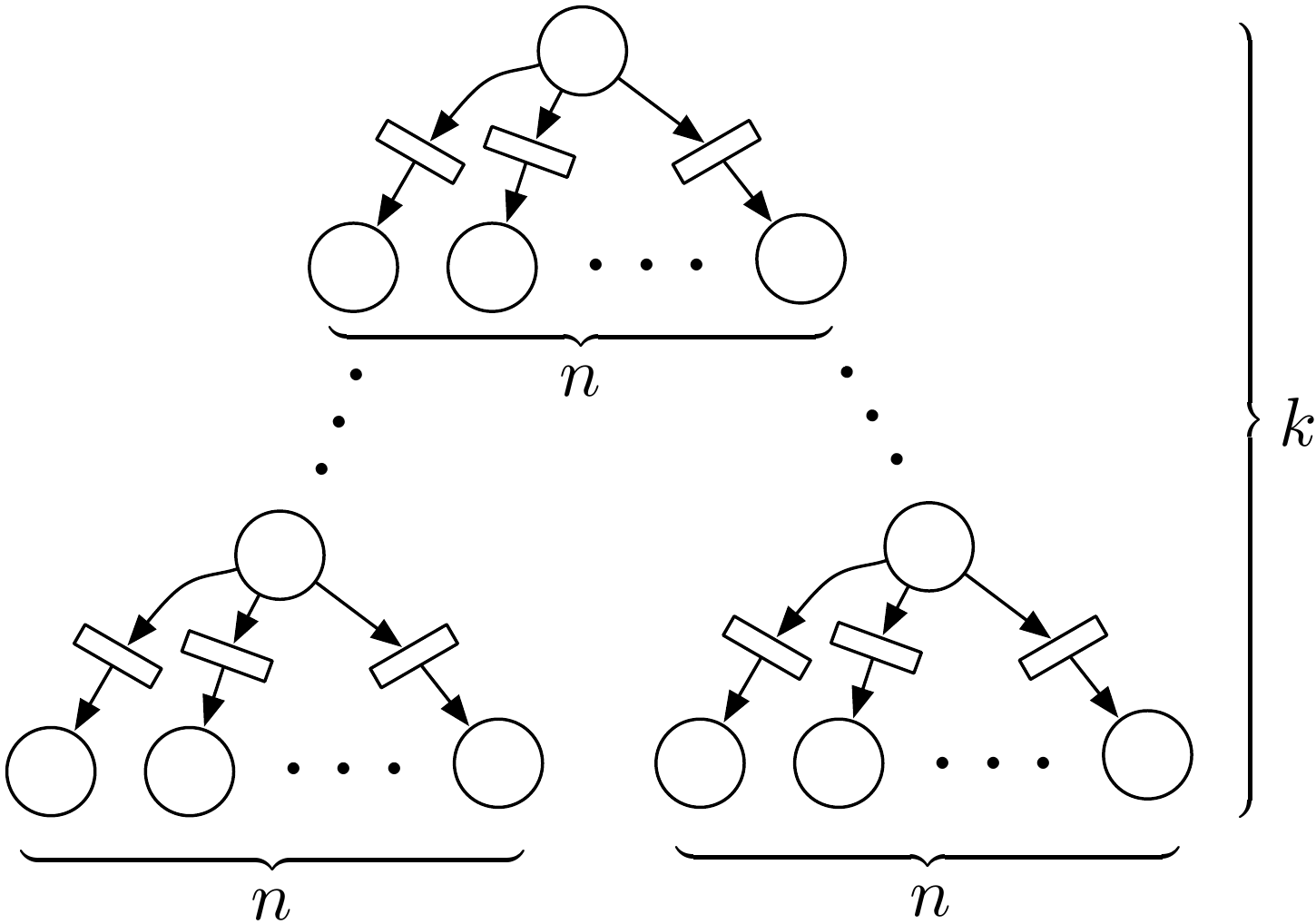}
        \caption{$T_{\ldiag}^{n,k}$ - separate transitions between parent and
        children.}
        \label{fig:tldiag}
    \end{subfigure}
    \caption{Complete tree nets of depth $k$ and width $n$.}
\end{figure}

\begin{example}
    \label{exm:tdiag}
    As an example of the use of the algebra of nets with boundaries, consider
    the net $T_\diag^{n,k}$, where $k,n\geq 1$, in \figref{fig:tdiag}.
    We can give a simple decomposition that relies on the components nets illustrated in~\figref{fig:Tnk}.
\begin{figure}
\begin{equation}
\lower12pt\hbox{$\includegraphics[height=1.2cm]{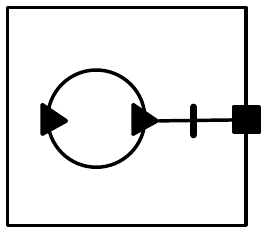}$} \qquad
\lower12pt\hbox{$\includegraphics[height=1.1cm]{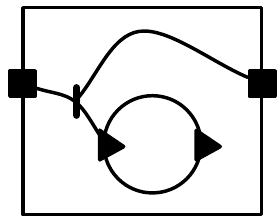}$} \qquad
\lower15pt\hbox{$\includegraphics[height=1.4cm]{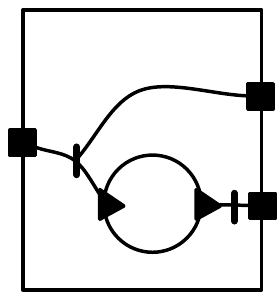}$} \qquad
\lower8pt\hbox{$\includegraphics[height=.7cm]{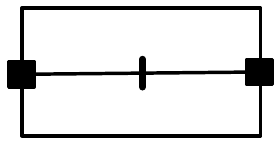}$} \qquad
\lower12pt\hbox{$\includegraphics[height=1.2cm]{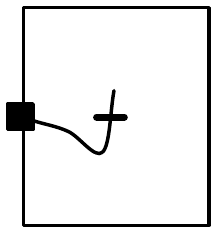}$}
\end{equation}
\[
R: 0\to 1
\qquad
L_\diag: 1\to 1
\qquad
N_\diag: 1\to 2
\qquad
I: 1\to 1
\qquad
\bot: 1\to 0
\]
\caption{Components used in the decomposition of $T_\diag^{n,k}$.\label{fig:Tnk}}
\end{figure}
First, we define the net with boundaries $B_\diag^{n,k}:1\to 0$ by recursion on $k$:
\begin{equation}\label{eq:bdelta}
    B_\diag^{n, k} \Defeq \begin{cases}
            {L_\diag}^n \comp \rightEnd & \mbox{if } k = 1 \\
            {(N_\diag \comp ( I \otimes B_\diag^{n,i} ) )}^n \comp \rightEnd & \mbox{if } k = i + 1 \\
          \end{cases}
\end{equation}
whence it follows that \vspace{-1.2em}
\begin{equation}\label{eq:tdiag}
T_\diag^{n,k} \cong R \comp B_\diag^{n,k}.
\end{equation}
\end{example}
The decomposition of $T_\diag^{2,2}$, following the definition
in~\eqref{eq:tdiag}, is illustrated in \figref{fig:t22}; components enclosed
with $\includegraphics[width=1cm]{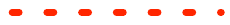}$ are composed with `$\comp$', while
components enclosed with $\includegraphics[width=1cm]{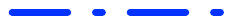}$ are composed
with `$\otimes$'.

\begin{figure}
\[
\!\includegraphics[height=2.9cm]{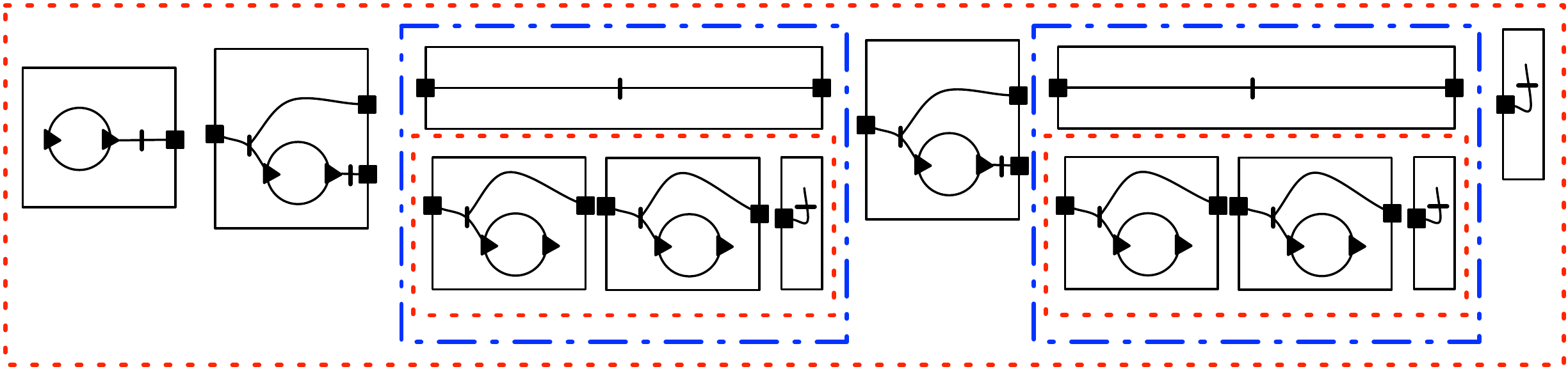}
\]
\caption{Decomposition of $T_\diag(2,2)$.\label{fig:t22}}
\end{figure}

\section{Wiring Decompositions}
\label{sec:width}


To formalise the decomposition of nets with boundaries, such as that presented
in \exmref{exm:tdiag}, we introduce the concept of a \emph{wiring
decomposition}. A \emph{wiring expression} is a syntactic term formed from the
following grammar:
\[
    T \bnfEq x \bnfSep T \comp T \bnfSep T \otimes T
\]
that is, a binary tree, with internal `$\comp$' and `$\otimes$' nodes and
variables at the leaves.

A \emph{variable assignment} $\wiringVarEnv$ is a map that takes variables to
nets with boundaries. Given a pair $\wiringDecomp$ of a wiring expression
$t$ and variable assignment $\wiringVarEnv$, its semantics
$\semanticsOf{t}_{\wiringVarEnv}$ is a net with boundaries, defined inductively:
\[
\semanticsOf{x}_\wiringVarEnv \Defeq \wiringVarEnv(x)
\qquad
\semanticsOf{t_1 \comp t_2}_\wiringVarEnv \Defeq
        \semanticsOf{t_1}_\wiringVarEnv \comp \semanticsOf{t_2}_\wiringVarEnv
\qquad
\semanticsOf{t_1\otimes t_2}_\wiringVarEnv \Defeq
        \semanticsOf{t_1}_\wiringVarEnv\otimes\semanticsOf{t_2}_\wiringVarEnv
\]
We implicitly assume that variable assignments are compatible with $t$: in the
sense that only nets with a common boundary are composed; we omit the details,
which are straightforward.

\begin{definition}
Given a net $N:k\to l$, we say that the pair $\wiringDecomp$ is a
\emph{wiring decomposition} of $N$ if $\semanticsOf{t}_\wiringVarEnv\cong N$.
\end{definition}

\begin{example}\label{exm:tdiag_decomp}
    A wiring decomposition of $T_\diag^{n,k}$
    can be obtained
    from~\eqref{eq:bdelta} and~\eqref{eq:tdiag} above
    by rewriting the equations as syntactic terms, with variables
    in place of each of the small component nets, and choosing a particular
    association for the `$\comp$' and `$\otimes$' expressions. We will see below
    that this particular choice of associativity is unimportant in terms of
    decomposition size but nonetheless has ramifications for the efficiency
    of our reachability checking algorithm (see \cite{Sobocinski2013} for examples).

\end{example}


\subsection{Reachability via Compositionality}
\label{sec:recap}


In this section we give a summary of the approach introduced
in~\cite{Sobocinski2013}, where, given a 1-bounded Petri net,
we decompose it using algebra of nets with boundaries to calculate
reachability in divide-and-conquer style. However, the technique is only viable for nets for which we can find ``small'' decompositions.

As discussed in \secref{sec:nets_with_boundaries}, each net with boundaries
determines an LTS, witnessing its step semantics. For a given reachability
problem, we can transform the LTS into a NFA, by letting the initial and final
states of the NFA be those corresponding to the initial and final markings.
Reachability then coincides with non-emptiness of the NFA's language.
To achieve a bounded statespace using our technique, we require that the
considered nets admit ``small'' decompositions (the precise definition of which
is presented \secref{sec:decomposition_width}.) 

We rely on the compositionality of nets with boundaries in order to perform
local checking of global reachability, w.r.t. interactions on a components'
boundaries: the NFA of a component net encodes the required ``protocol'' that
the net must engage in with its environment in order to reach a (locally) final
marking. Thus, to generate $NFA(x \comp y)$, it suffices to generate $NFA(x)$,
and $NFA(y)$ and compose them using a variant of the product construction:
$(a,b)\dtrans{\alpha}{\beta}(a’,b’)$ iff $\exists \gamma.\;
a\dtrans{\alpha}{\gamma}a’ \mathrel{\wedge}
b\dtrans{\gamma}{\beta}b’$\footnote{Similarly, we can perform
$\ten$-composition on NFAs with a different modification of the standard
product construction.}.

Hiding internal computations improves the performance of our technique; we
perform $\epsilon$-closure on the obtained NFAs, identifying internal states
that are distinguished only by transitions that do not alter the net's
protocol. Further, we avoid state explosion by minimising the NFA's
representation size, applying determinisation followed by DFA-minimisation to
generate an automaton that recognises the same language, but with potentially
simpler structure. Observe that after performing $\epsilon$-closure and
minimisation on the NFA of a net $N:0\to 0$ we have either the trivially
accepting, or trivially rejecting automaton.

Furthermore, many nets have a repeated internal structure---several examples
being presented in~\cite{Sobocinski2013}, and this paper. By exposing this
repeated structure through decomposition, we avoid duplicating work, by
employing memoisation such that conversion to NFA, or NFA composition is only
performed once.

\begin{example}
Consider a decomposition of $T_\diag^{2,2}$, as defined in~\eqref{eq:bdelta}
and~\eqref{eq:tdiag}, and illustrated in \figref{fig:t22}. Let the initial
marking be a single token at the root place, and the final marking having only
leaves marked. The minimal DFAs obtained from this decomposition are presented
in \figref{fig:tdiag_nfas}\footnote{We have omitted error states if present.
Labels indicate interaction with the boundaries: 00/1 is action on the right
boundary, with no action on either left boundary. `$\ast$' means either 0 or
1.}. For example, observe that $B_\diag^{2,1}$ reaches its local accept state
upon interacting once on its left boundary. Reachability is confirmed: the
minimal DFA representing $T_\diag^{2,2}$ is the trivial accepting automaton.
\end{example}
\begin{figure}[h]
    \centering
    \begin{subfigure}{0.15\textwidth}
        \centering
        \includegraphics[scale=0.5]{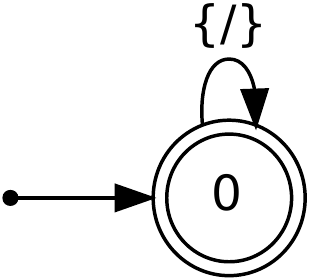}
        \caption{$T_\diag^{2,2}$}
    \end{subfigure}
    \hfill
    \begin{subfigure}{0.3\textwidth}
        \centering
        \includegraphics[scale=0.5]{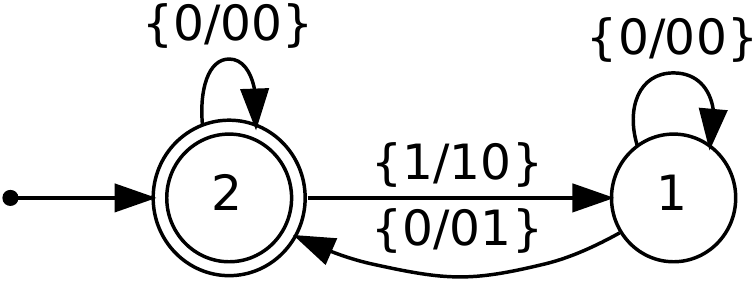}
        \caption{$N_\diag$}
    \end{subfigure}
    \hfill
    \begin{subfigure}{0.4\textwidth}
        \centering
        \includegraphics[scale=0.5]{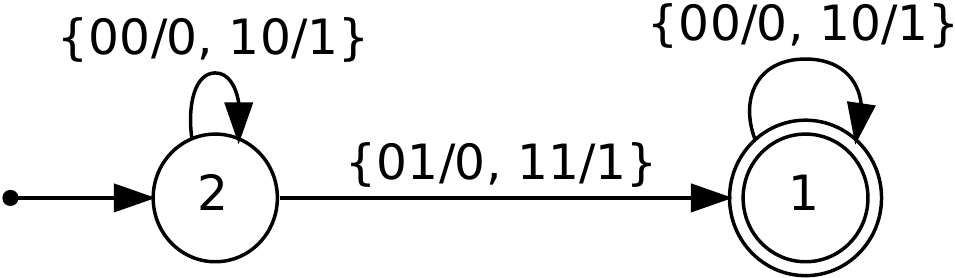}
        \caption{$I \ten B_\diag^{2,1}$}
    \end{subfigure}
    \\[0.3cm]
    \begin{subfigure}{0.15\textwidth}
        \centering
        \includegraphics[scale=0.5]{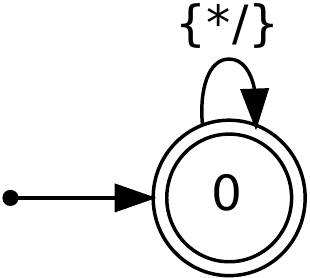}
        \caption{$\bot$}
    \end{subfigure}
    \hfill
    \begin{subfigure}{0.3\textwidth}
        \centering
        \includegraphics[scale=0.5]{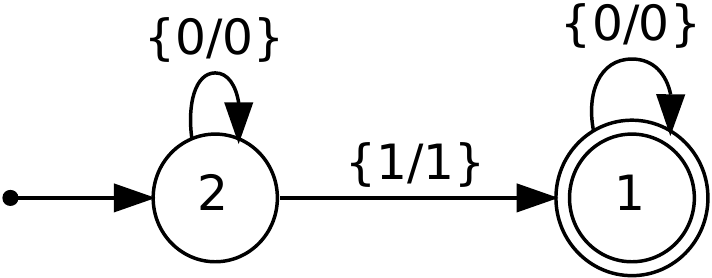}
        \caption{$L_\diag$}
    \end{subfigure}
    \hfill
    \begin{subfigure}{0.4\textwidth}
        \centering
        \includegraphics[scale=0.5]{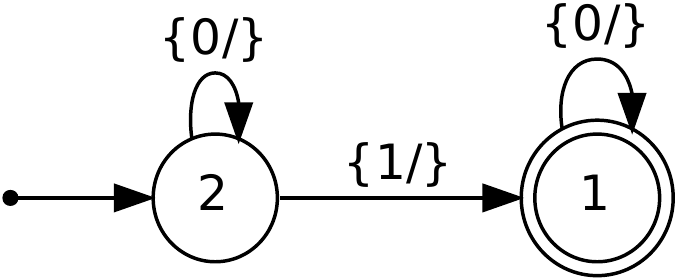}
        \caption{$B_\diag^{2,1}$}
    \end{subfigure}
    \caption{Component NFAs of the right-associative $T_\diag^{2,2}$
    decomposition.}
    \label{fig:tdiag_nfas}
\end{figure}

\subsection{Decomposition width}
\label{sec:decomposition_width}

As explained in the preceding section, the ``size'' of a decomposition is important for performance. We formalise this below.
\begin{definition}[Decomposition width]\label{defn:width}
    We say that a wiring decomposition, $\wiringDecomp$, of a net with
    boundaries has width $k \in \N$, if:
    \begin{enumerate}[(i)]
    \item $\forall x \in t$, $\semanticsOf{x}_{\wiringVarEnv} : l \to r$, with
        places $P$, satisfies $max(l,\left|P\right|,r) \leq k$, and
    \item for all subexpressions $t'$ of $t$, if
        $\semanticsOf{t'}_{\wiringVarEnv} : l \to r$ then $max(l,r) \leq k$.
    \end{enumerate}
    A net has \emph{decomposition width} $k$ if it has a wiring decomposition
    of width $k$. A family of nets $\setof{N_i}_{i \in I}$ has
    \emph{bounded decomposition width} if there exists $k \in \N$ such that for all
    $i \in I$, $N_i$ has decomposition width $k$.
\end{definition}

\begin{lemma}[Invariance w.r.t. associativity]
\label{lem:associativity_invariance}
Given a wiring decomposition, $\wiringDecomp$, of a net $N: l \to r$ that has width $k$,
and given a wiring expression $t'$ such that $t'$ is equivalent to $t$ up to
associativity of `$\comp$' and `$\ten$' then $(t', \wiringVarEnv)$ also has
width $k$ and $\semanticsOf{t'}_{\wiringVarEnv} : l \to r$.
\end{lemma}
\begin{proof}
Write $t \sim t'$ for equivalence up to associativity and proceed by induction on the structure
of $t'$. If $t'$ is a variable then it is equal to $t$ and hence the result follows.

Suppose that $t'$ is an n-fold `$\ten$'-composition of some $t_i'$ for $1 \leq i \leq n$
such that $t$ is also an n-fold `$\ten$'-composition of some $t_i$ with any other possible
association with $t_i' \sim t_i$.
By the induction hypothesis we see that
each $(t_i',\wiringVarEnv)$
has width $k$ and
$\semanticsOf{t_i'}_{\wiringVarEnv} : l_i \to r_i$ where $l = (\sum_{1 \leq i \leq n} l_i) \leq k$
and $r = (\sum_{1 \leq i  \leq n} r_i) \leq k$.
Any subexpression of $t'$ is either a subexpression of one of the $t_i'$ (and hence satisfies boundedness)
or some expression $t''$
containing a `$\ten$'-composition of a subsequence $I$ of the $t_i'$. The boundaries of $(t'',\wiringVarEnv)$ have
size $l_I = \sum_I l_i \leq k$ and $r_I = \sum_I r_i \leq k$.
Hence $(t',\wiringVarEnv)$ also has width $k$ and $\semanticsOf{t'}_{\wiringVarEnv} : l \to r$
as required. \qed
\end{proof}

Note that the algebra of nets with boundaries is actually an algebra of
directed hypergraphs (that happens to be compositional w.r.t. the net
semantics). Thus, the notion of decomposition width, introduced above,
is---more generally---a structural property of directed hypergraphs.

\begin{example}
Consider the net $T_\diag^{n,k}$ from~\figref{fig:tdiag}, decomposed
in~\eqref{eq:bdelta} and~\eqref{eq:tdiag}. For any $n, k$, this wiring
decomposition has width 2: observe that every component net of $B_\diag^{n,k}$
has at most one place and two boundary ports on either side Furthermore, it is
easy to confirm that at each internal node of the tree, two subtree nets are
composed such that the resulting net has boundaries $\leq 2$, i.e.
subexpressions have boundaries $\leq 2$. That is, a decomposition width of 2.
\end{example}

\section{Harnessing the full algebra}
\label{sec:extended_algebra}



In this section we use the full algebra of nets with boundaries in order to obtain decompositions of bounded width nets that do not have satisfactory decompositions using merely the subalgebra used in~\cite{Sobocinski2013}, described in Remark~\ref{rmk:restriction}.
Since, as explained in~\secref{sec:recap}, a bounded decomposition width is a necessary condition for the applicability of our reachability checking approach,  by doing so, we are able to extend its applicability to several natural families of nets.


\begin{example}\label{exm:tldiag_decomp}
    \label{exm:tldiag}
    Consider the family of nets $T_\ldiag^{n,k}$ in \figref{fig:tldiag}. These
    nets are similar to those discussed in \exmref{exm:tdiag}, but with $n$
    \emph{distinct} transitions from any non-leaf node to its children.



There is no way of obtaining a decomposition of bounded width with the
restriction of Remark~\ref{rmk:restriction}, i.e. at most one transition
connected to each boundary port. To see why, assume we have a decomposition and
consider the component that contains the root node: as we increase $n$ one
would have to either increase the size of the boundary or increase the number
of places within the component. Without the restriction we can connect multiple transitions to the same boundary port, and so modify the construction of \exmref{exm:tdiag} to obtain a decomposition for $T_\ldiag^{n,k}$: Again, first define the component net $B_\ldiag^{n,k}:1\to 0$ by recursion on $k$:
\begin{equation}\label{eq:bldiag}
    B_\ldiag^{n, k} \Defeq \begin{cases}
            {L_\ldiag}^n \comp \rightEnd & \mbox{if } k = 1 \\
            {(N_\ldiag \comp ( I \otimes B_\ldiag^{n,i} ) )}^n \comp \rightEnd & \mbox{if } k = i + 1 \\
          \end{cases}
\end{equation}
whence we have that:\vspace{-1.3em}
\begin{equation}\label{eq:tldiag}
T_\ldiag^{n,k} \cong R \comp B_\ldiag^{n,k}.
\end{equation}
\end{example}

\begin{figure}[h]
\vspace{-2.5em}
\begin{equation}
    \lower15pt\hbox{$\includegraphics[height=1.2cm]{Rnet}$} \qquad
    \lower15pt\hbox{$\includegraphics[height=1.2cm]{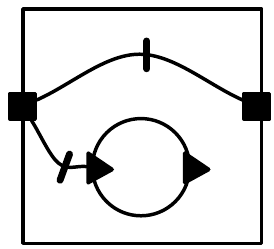}$} \qquad
    \lower19pt\hbox{$\includegraphics[height=1.4cm]{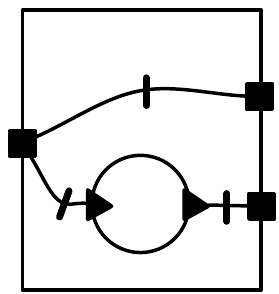}$} \qquad
    \lower10pt\hbox{$\includegraphics[height=.7cm]{Inet}$} \qquad
    \lower15pt\hbox{$\includegraphics[height=1.2cm]{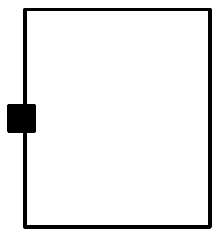}$}
\end{equation}
\[
    R: 0\to 1
\qquad
    L_\ldiag: 1\to 1
\qquad
    N_\ldiag: 1\to 2
\qquad
    I: 1\to 1
\qquad
    \rzero: 1\to 0
\]
\caption{Components used in the decomposition of $T_\ldiag^{n,k}$.\label{fig:Tlnk}}
\end{figure}


In addition to the decompositions in Examples~\ref{exm:tdiag_decomp}
and~\ref{exm:tldiag_decomp} we will consider two other families of nets that
are ``densely'' connected and show that they nevertheless have bounded
decomposition width.
\begin{figure}
\vspace{-2.5em}
\begin{subfigure}[b]{.35\textwidth}
\[
\includegraphics[height=3cm]{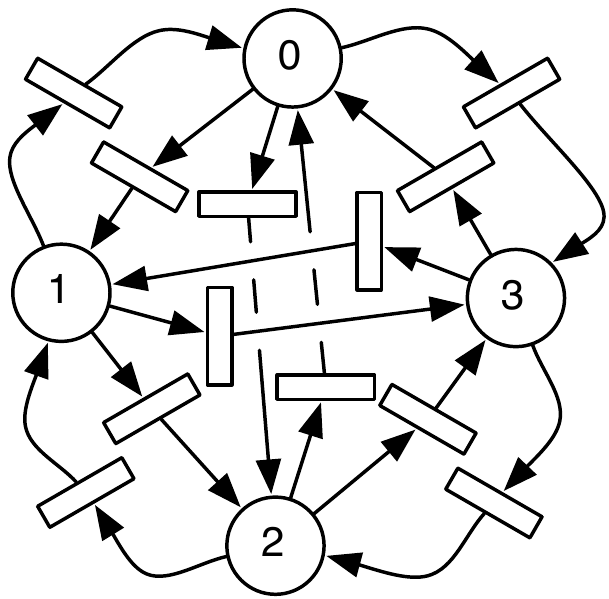}
\]
\caption{Net $C_4$.\label{fig:clique4}\vspace{-.3cm}}
\end{subfigure}
\begin{subfigure}[b]{.55\textwidth}
\[
\includegraphics[height=2cm]{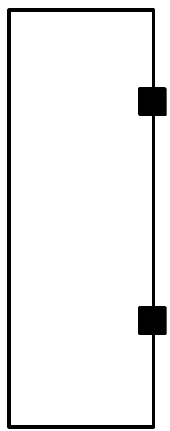}
\qquad
\includegraphics[height=2cm]{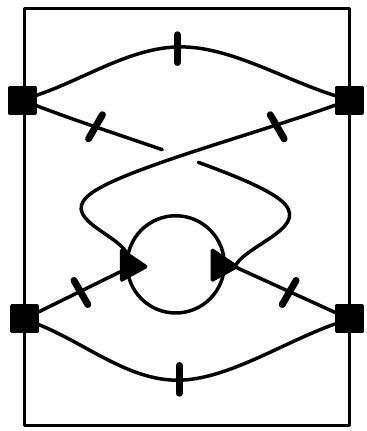}
\qquad
\includegraphics[height=2cm]{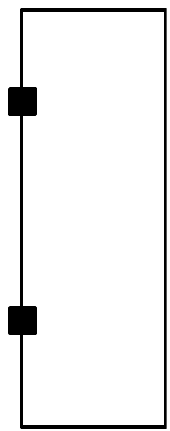}
\]
\[
\lzero \otimes \lzero: 0 \to 2
\qquad
S: 2\to 2
\qquad
\rzero \otimes \rzero: 2 \to 0
\]
\caption{$C_n$ wiring decomposition components.\label{fig:cliqueComponents}\vspace{-.1cm}}
\end{subfigure}
\caption{Decomposing cliques.}
\end{figure}

\begin{example}\label{exm:clique_decomp}
Consider the clique net $C_n$: it has $n$ places and $n\times (n - 1)$
transitions, one from each place to every other. An illustration of $C_4$ is
given in \figref{fig:clique4}. It is easy to see that the flowgraph of $C_n$
has treewidth $n-1$, on the other hand $C_n$ has decomposition width 2 for any
$n$.

The decomposition is simple and uses the components illustrated
in~\figref{fig:cliqueComponents}. Indeed, it is not difficult to see that $C_n
\cong (\lzero \otimes \lzero) \comp S^{n} \comp (\rzero \otimes \rzero)$.
\end{example}

\begin{figure}
\vspace{-2.5em}
\begin{subfigure}[b]{.35\textwidth}
\[
\lower20pt\hbox{$\includegraphics[height=3.3cm]{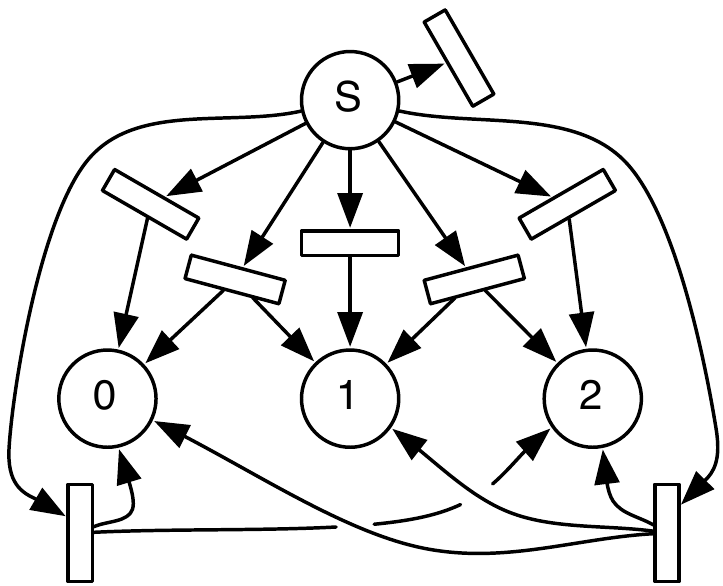}$}
\]
\caption{Net $P_3$.\label{fig:subset_net}\vspace{-.3cm}}
\end{subfigure}
\begin{subfigure}[b]{.55\textwidth}
\[
\includegraphics[height=1.5cm]{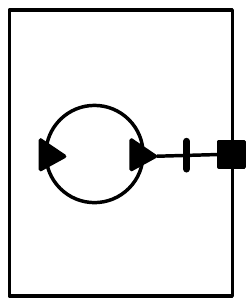}
\qquad
\includegraphics[height=1.5cm]{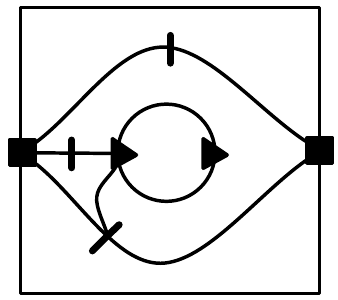}
\qquad
\includegraphics[height=1.5cm]{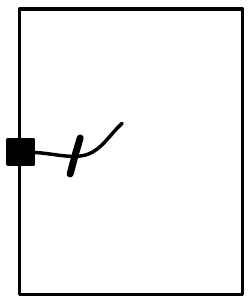}
\]
\[
R: 0 \to 1 \qquad
P: 1 \to 1 \qquad
\bot: 1 \to 0
\]
\caption{$P_i$ wiring decomposition components.\label{fig:Pncomponents}\vspace{-.1cm}}
\end{subfigure}
\caption{Decomposing subset nets.}
\end{figure}

\begin{example}\label{exm:subset_decomp}
Consider the net $P_n$, $n\geq 0$, with $n+1$ places. There is a chosen place
$S$, with the remaining places $0,1,\dots,n-1$, the elements of $[n]$. There
are $2^n$ transitions in $P_n$, all with the single source $S$ and targets
the elements of $2^{\ordinal{n}}$. See~\figref{fig:subset_net} for an
illustration of $P_3$. For any $n>1$, $P_n$ has a wiring decomposition of width
$1$: indeed, consider the components in~\figref{fig:Pncomponents}, then an easy
calculation confirms that $P_n \cong R \comp P^n \comp \bot$.
\end{example}

Having extended the scope of our reachability technique to that of the full
algebra of nets with boundaries, we are able to handle more examples, such as
those presented in this section.

\section{Principles of decomposition}
\label{sec:principles}


In Examples~\ref{exm:tdiag_decomp}, \ref{exm:tldiag_decomp},
\ref{exm:clique_decomp} and \ref{exm:subset_decomp} we exhibited several
families of nets with bounded decomposition width. In this section we develop
the theory of decompositions that will allow us to place lower bounds on the
size of shared boundaries in certain decompositions. Taking these initial
observations into consideration, we conjecture that the family of grid nets
$\{G_n\}_{n\in\N_+}$, with $G_3$ illustrated in~\figref{fig:grid}, does not
have bounded decomposition width.
\begin{figure}[h]
\centering
\vspace{-1em}
\begin{subfigure}[b]{.30\textwidth}
\centering
\[
\includegraphics[height=3.7cm]{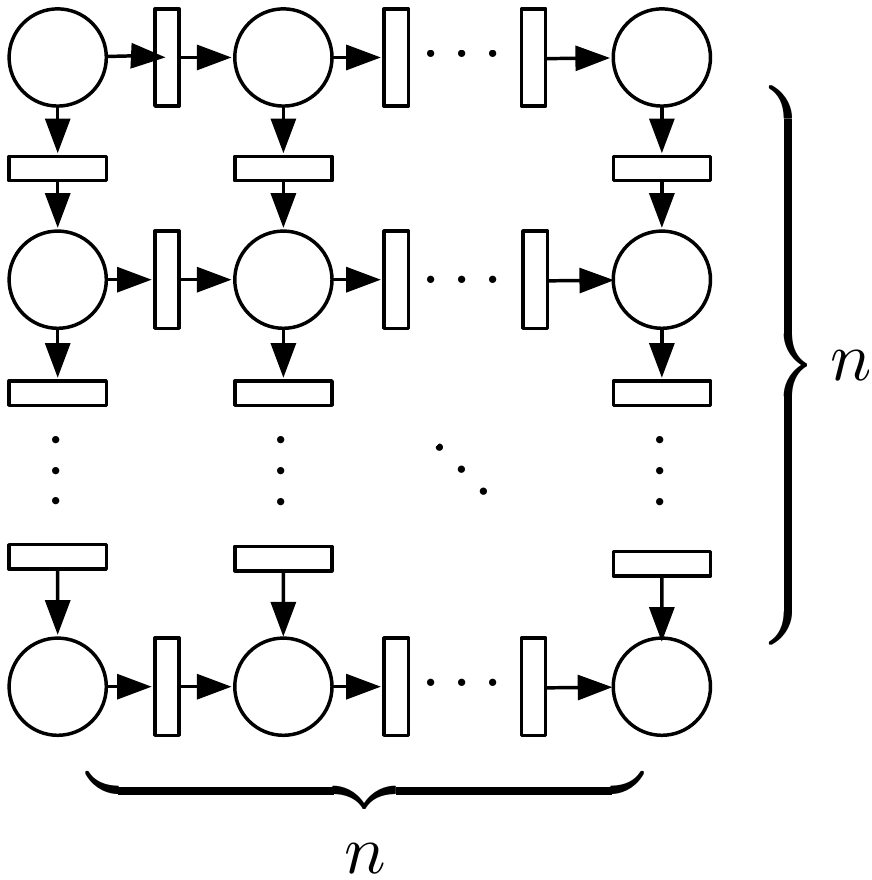}
\]
\caption{$G_n$.\label{fig:Gn}\vspace{-.3cm}}
\end{subfigure}
\hfill
\begin{subfigure}[b]{.6\textwidth}
\centering
\[
\raise15pt\hbox{$\includegraphics[height=2.9cm]{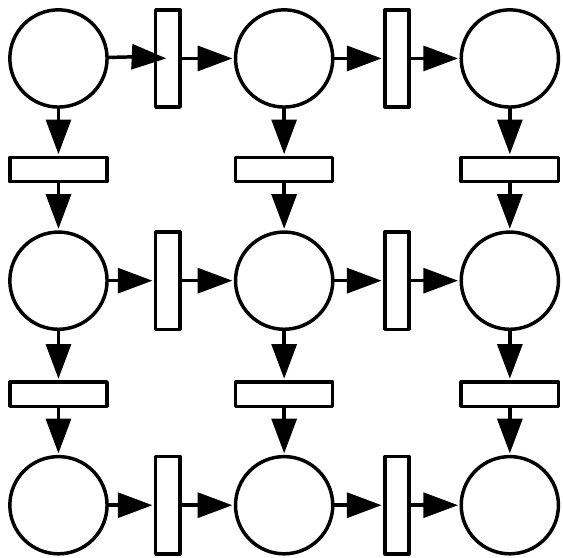}$}
\qquad
\includegraphics[height=3.7cm]{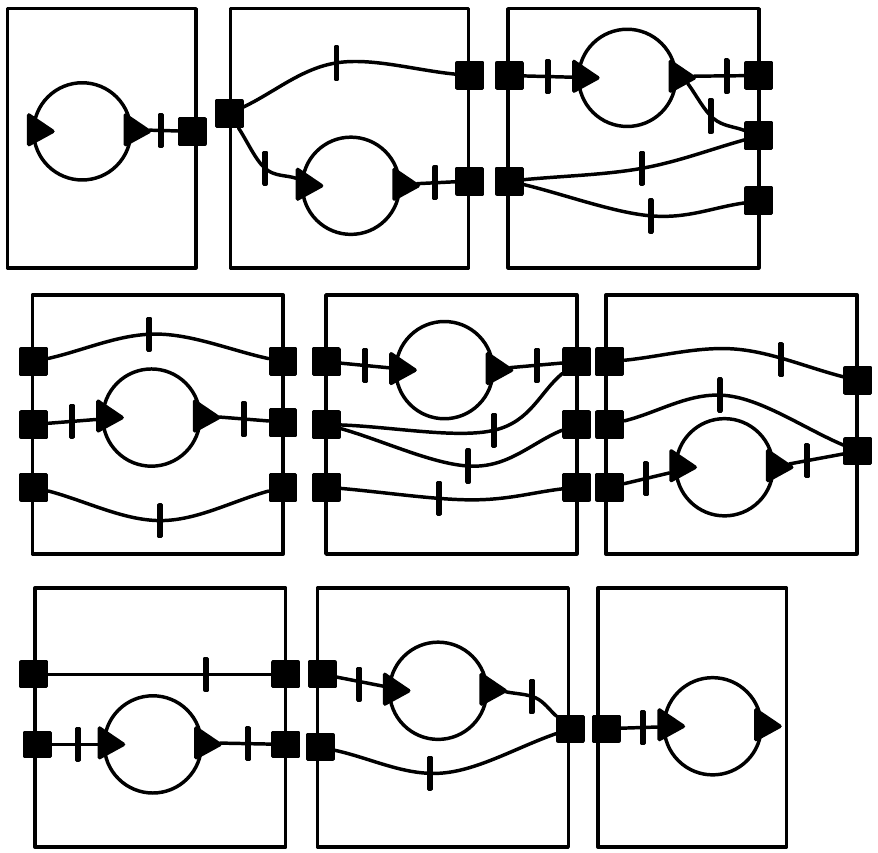}
\]
\caption{$G_3$ and a decomposition of width 3.\label{fig:G3}\vspace{-.3cm}}
\end{subfigure}
\caption{Decomposing grids.\label{fig:grid}\vspace{-.5cm}}
\end{figure}

\subsection{Portsets, connections and networks}
For a net $N:k\to l$ and $P\subseteq \places{N}$, the set of \emph{place ports}
of $P$ is: $\ports{P} \Defeq \setof{ \inport{p} \mid p \in P } \cup \setof{
\outport{p} \mid p \in P }$\footnote{For the sake of concreteness,
$\inport{p}\Defeq (p,in)$, $\outport{p}\Defeq (p,out)$.}. When we refer to
$N$'s \emph{boundary ports}, we mean the elements of $\ordinal{k}+\ordinal{l}$.
When referring to individual boundary ports we will write $\leftb{i}$ for
$(i,0)$ and $\rightb{i}$ for $(i,1)$. The set of ports of $N$ is all its place
ports and boundary ports: $\ports{N}\Defeq\ports{\places{N}}\cup
(\ordinal{k}+\ordinal{l})$. We will usually refer to sets of ports as
\emph{portsets}. Given a transition $t$, the portset of $t$ is:
\[
\ports{t} \Defeq \lbrace \outport{p} \mid p \in \pre{t} \rbrace \union \lbrace
\inport{p} \mid p \in \post{t} \rbrace \union (\source{t} + \target{t}).
\]
We will usually write portsets using angle brackets. For instance, consider the
net $R:0\to 1$ in \figref{fig:Pncomponents}, with $\places{R}=\{p\}$ and
$\trans{R}=\{t\}$. Then $\ports{\{p\}}=\portset{ \inport{p},\, \outport{p} }$,
$\ports{R} = \portset{ \inport{p},\, \outport{p},\, \rightb{0} }$ and
$\ports{t} = \portset{ \outport{p},\, \rightb{0} }$.

We will refer to sets of portsets as a \emph{connections} and write them using
square brackets. The \emph{connection} of a port $p\in\ports{N}$
is the set of portsets of all transitions that connect to $p$:
\[
\conn{p} \Defeq \setof{\ports{t}\backslash\{p\} \mid t \in \trans{N} \wedge \{
p \} \subset \ports{t}}.
\]
For example, in $P:1\to 1$ in \figref{fig:Pncomponents}:
$\conn{\leftb{0}}=\footprint{ \portset{\rightb{0}},\, \portset{\inport{p}},\,
\portset{\inport{p},\, \rightb{0}} }$, $\conn{\inport{p}}=\footprint{
\portset{\leftb{0}},\, \portset{\leftb{0},\, \rightb{0}} }$,
$\conn{\outport{o}}=\varnothing$ and $\conn{\rightb{0}}=\footprint{
\portset{\leftb{0},\,\inport{p}},\portset{\leftb{0}}}$.

We will find it useful to sometimes restrict $\conn{p}$ to those sets of ports
that intersect non-trivially with some subset $R$ of the ports of a net. We
write:
\[
\connr{R}{p} = \setof{ K \cap R \mid K\in\conn{p}, K\cap R\neq \varnothing }.
\]

Suppose that $N:k\to l$ is a net with boundaries. An \emph{oriented partition} is $\fat{P} = (P_l, P_r)$, where $\{P_l,P_r\}$ is a partition of $\places{N}$ and $P_l,P_r\neq\varnothing$. Given an oriented partition, we define the \emph{extended ports} of $P_l$ and $P_r$: $\eports{P_l} \Defeq
\ports{P_l} \cup \setof{ (i, 0) \mid i < k }$ and $\eports{P_r} \Defeq
\ports{P_r} \cup \setof{ (i, 1) \mid i < l }$. These contain the ports of the
places in each set and boundary ports: $P_l$ from the left boundary and $P_r$ from the right boundary.

Given an oriented partition $\fat{P}$ of a net $N$, we need to express how the places in the two disjoint place sets are interconnected. We will refer to sets of connections as \emph{networks}. Then the network from $P_l$ to $P_r$ consists of the connections to extended ports of $P_r$, for each extended port of $P_l$:
\[
\network{P_r}{P_l} \Defeq
\left\lbrace
\connr{\eports{P_r}}{p} \mid p \in \eports{P_l}
\right\rbrace
\]
and similarly for $\network{P_l}{P_r}$.




\begin{example}
Consider the clique $C_4:0\to 0$, illustrated in
\figref{fig:clique4}, and the oriented partition $\fat{P} = (\setof{0,1}, \setof{2, 3})$. Then:
\begin{multline*}
\connr{\setof{2,3}}{\outport{0}} = \footprint{\portset{\inport{2}}, \portset{\inport{3}}}
    = \connr{\setof{2,3}}{\outport{1}}, \\
\connr{\setof{2,3}}{\inport{0}} = \footprint{\portset{\outport{2}}, \portset{\outport{3}}}
    = \connr{\setof{2,3}}{\inport{1}}.
\end{multline*}
Thus $\network{\setof{2,3}}{\setof{0,1}} =
    \setof{ \footprint{\portset{\outport{2}}, \portset{\outport{3}}},
    \footprint{\portset{\inport{2}}, \portset{\inport{3}}} }$
    and by a symmetric argument $\network{\setof{0,1}}{\setof{2,3}} =
    \setof{ \footprint{\portset{\outport{0}}, \portset{\outport{1}}},
    \footprint{\portset{\inport{0}}, \portset{\inport{1}}} }$. Note that, although cliques contain many transitions, the networks between partitions are small: in fact, it is not difficult to show that for all $n$, any oriented partition $(P_l,P_r)$ of the places of $C_n$ satisfies $|\network{P_l}{P_r}|=|\network{P_r}{P_l}|=2$.
    Roughly speaking, the amount of information to describe connections from one partition to another is constant, and this is the key insight that leads to the decompositions presented in Examples~\ref{exm:clique_decomp} and~\ref{exm:subset_decomp}.
%
\end{example}

\subsection{Bases, dimension and pure decompositions}

We now show that there is a general connection between the networks of an
oriented partition, and the internal boundary of any corresponding `$\comp$'
decompositions. First we introduce the notion of a \emph{basis} of a network:
\begin{definition}[Basis]
    Given a network $N$, a vector of connections $b_0 \dots b_{n - 1}$, is
    a \emph{basis} of $N$ iff $\forall c \in N$, there exists
    $l\subseteq\ordinal{n}$ with $c = \bigcup_{i \in l} b_i$. That is, every
    connection in $N$ can be written as the union of a subset of the
    connections of the basis. The dimension of a network $N$, $\dimension{N}$
    is the size of its smallest basis.
\end{definition}

%

\begin{figure}[h]
    \centering
    \begin{subfigure}{.49\textwidth}
        \includegraphics[height=2cm]{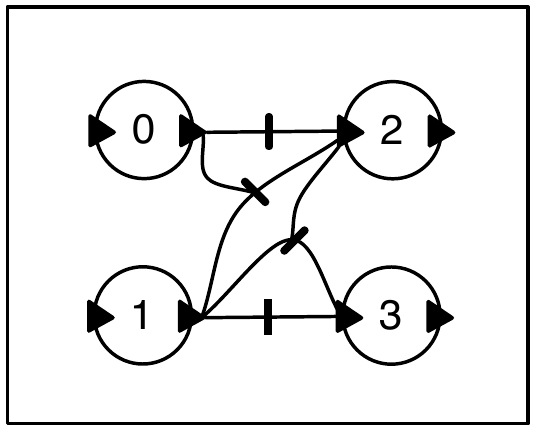}
        \includegraphics[height=2cm]{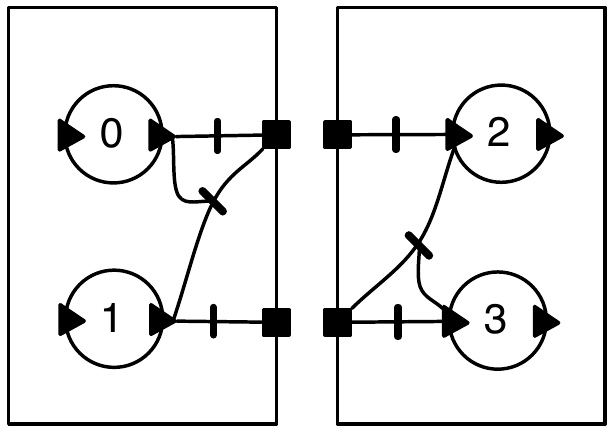}
    \caption{A pure composition.\label{fig:pure}}
    \end{subfigure}
    \begin{subfigure}{.49\textwidth}
	\includegraphics[height=2cm]{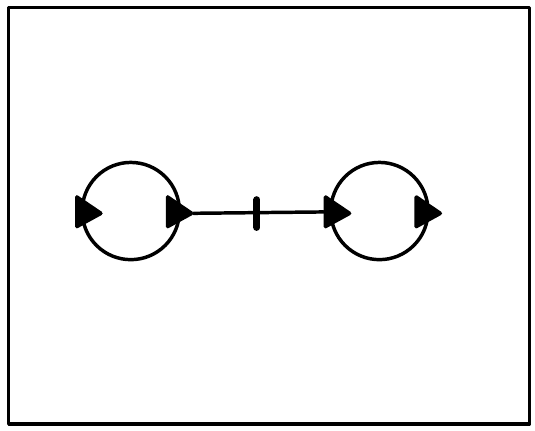}
	\includegraphics[height=2cm]{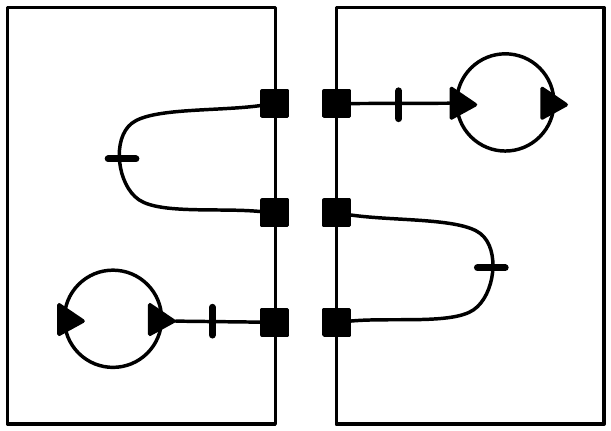}
	\caption{A non-pure composition.\label{fig:nonpure}}
	\end{subfigure}
	\caption{Composition examples.}
\end{figure}

Suppose we have a net $N:k\to l$ and a decomposition $N\cong N_l\comp N_r$ (*)
where $N_l:k\to n, N_r:n\to l$, with places $P_l$ and $P_r$, respectively.
Through slight abuse of notation we equate the places of $P_l$ and $P_r$ with
the corresponding places in $N$ by fixing a concrete isomorphism that witnesses
(*).  In particular we obtain an oriented partition $(P_l,P_r)$ of $N$.


The connections of each shared-boundary port $j<n$ to $N_l$ and $N_r$ are just
$\bconn{N_l}{j} \Defeq \connwithnet{N_l}{\rightb{j}}$ and $\bconn{N_r}{j}
\Defeq \connwithnet{N_r}{\leftb{j}}$ where the superscripts refer to the
ambient net in which the calculation takes place. We say that the composition
$N_l\comp N_r$ is \emph{pure} iff, for all $j<n$ no portset in $\bconn{N_l}{j}$
contains a right boundary port $\rightb{i}$ and, symmetrically,
no portset in $\bconn{N_r}{j}$ contains a left boundary port $\leftb{i}$.  In
other words, no transition in $N_l$ or $N_r$ connects to two different shared
boundary ports. It follows that in pure decompositions $\connr{N_l}{j}$ and
$\connr{N_r}{j}$ are connections in $N$. All examples of decompositions we have
considered so far are pure; a non-pure decomposition is illustrated
in~\figref{fig:nonpure}.


\begin{example}
Consider the net in \figref{fig:pure} and the corresponding pure
decomposition. The shared-boundary connections are as follows:
\begin{align*}
    \bconn{N_1}{0} &=
        \footprint{\portset{\outport{0}}, \portset{\outport{0}, \outport{1}}}
        &
    \bconn{N_2}{0} &=
        \footprint{\portset{\inport{2}}} \\
    \bconn{N_1}{1} &=
        \footprint{\portset{\outport{1}}}
        &
    \bconn{N_2}{1} &=
        \footprint{\portset{\inport{2}, \inport{3}}, \portset{\inport{3}}}
\end{align*}
\end{example}

\begin{proposition}\label{pro:decomposition_basis}
    \label{lem:basis_from_boundary_footprints}
    Given a net $N:k\to l$ together with a pure decomposition $N_l:k\to n,
    N_r:n\to l$, the vector $\left(\bconn{N_r}{i}\right)_{i<n}$ is a basis for
    $\network{N_r}{N_l}$, and $(\bconn{N_l}{i})_{i<n}$ is a basis for
    $\network{N_l}{N_r}$.
\end{proposition}
\begin{proof}
The purity of the composition implies that all transitions in the composition
(minimal synchronisations) are of the form $(\{u\},\{v\})$, $u\in\trans{N_1}$,
$v\in\trans{N_2}$, where $\target{u}=\source{v}$, a single shared-boundary
port. Then, it follows that for each $p\in \eports{P_l}$:
\begin{align*}
& \connr{\eports{P_r}}{p}= \\
& \setof{ \ports{t}\cap\eports{P_r} \mid t\in \trans{N}, p\in \ports{t}, \ports{t}\cap\eports{P_r}\neq \varnothing} \\
=& \setof{ \ports{v} \mid v\in \trans{N_r}, \exists u\in \trans{N_l}.\, p\in \ports{u} \wedge \target{u}=\source{v} } \\
 =& \bigcup_{\setof{i \;\mid\; \exists u\in N_l.\, p\in\ports{u}, \target{p}=i}} \bconn{N_r}{i}.
\end{align*}
The second case follows by symmetry.\qed
\end{proof}

Proposition~\ref{pro:decomposition_basis} leads to the following immediate corollary.
\begin{corollary}\label{cor:no_basis_no_decomp}
Suppose $N:k\to l$ decomposes into $N_1;N_2$ where $N_1:k\to n$, $N_2:n\to l$.
Suppose that $\fat{P} = (P_1, P_2)$ is the corresponding oriented partition.
Then $n\geq \dimension{\network{P_2}{P_1}}$.\qed
\end{corollary}


\begin{example}
Consider again the net in~\figref{fig:pure}. We have
\[
\network{\setof{2,3}}{\setof{0,1}}
= \setof{ \footprint{ \portset{ \inport{2} } }, \footprint{\portset{\inport{3}},\portset{\inport{2},\inport{3}}} }
\]
It is not difficult to see that a basis of size 1 does not exist, so there is no pure decomposition into nets with places $\setof{0,1}$, $\setof{2,3}$ with size 1 boundary.
\end{example}

Returning to the family of grid nets $G_n$ of~\figref{fig:grid}, for any $k\in
\N_+$, $G_k$ has a pure decomposition of width $k$; we illustrate this for
$G_3$ in~\figref{fig:G3}, and it is not difficult to generalise the
construction to arbitrary $k$. We omit the details here. We believe that
decompositions of size $<k$ do not exist: essentially if one constructs a grid
incrementally with pieces of size $<k$ one reaches a composition with boundary
$>k$, using an argument similar to the statement of
Corollary~\ref{cor:no_basis_no_decomp}.
\begin{example}
Consider $G_3$ in~\figref{fig:G3}. We can show that there is no pure
`$\comp$'-decomposition of width $< 3$. Clearly we can asssume that leaves each
have fewer than $2$ places. Using the conclusion of
Corollary~\ref{cor:no_basis_no_decomp}, we can show (by inspection) that for
every ``increasing'' sequence of partitions of the places of $G_3$,
$(P_{l,1},P_{r,1}), (P_{l,2},P_{r,2}), \dots, (P_{l,k},P_{r,k})$, where
$|P_{l,1}|,|P_{r,k}|<3$ and for all $1\leq i \leq k-1$, $P_{l,i}\subseteq
P_{l,i+1}$ and $|P_{l,i+1}\backslash P_{l,i}| < 3$, there exists $i$ such that
any composition $N_{l_i}:0\to n$, $N_{r_i}:n\to 0$ implies $n\geq 3$. We omit
the tedious details. It is also not difficult to extend this argument to
arbitrary pure decompositions (ie those that also have $`\otimes$' nodes).
\end{example}
The theory of general grid partitioning is
non-trivial (see, e.g.~\cite{Donaldson2000} for a pleasant overview) and we
leave the study of this conjecture for future work.
\begin{conjecture}
The family $\{G_n\}_{n\in\N_+}$ of~\figref{fig:grid} does not have bounded
decomposition width.
\end{conjecture}

\section{Conclusions and future work}
\label{sec:discussion}

We have considered the decomposition of 1-bounded Petri nets, employing the
full algebra of nets with boundaries. Through several examples we have
demonstrated that by doing so, we extend the applicability of our
divide-and-conquer algorithm for reachability checking. We have introduced and
examined the structural property of decomposition width for nets, and more
generally, directed hypergraphs.  Finally, we have developed the theory of
wiring decompositions to give a lower bound on the boundary size of certain
compositions.

Low decomposition width is not sufficient for avoiding state
explosion when generating the transition systems from nets---this, instead, is
the `semantic' property referred to in the Introduction. In future work, we
will consider this property, aiming to characterise the class of nets on which
our technique for reachability checking is viable. Here we have concentrated on
the necessary structural condition of (low) decomposition width, which also
deserves further study in its own right, and how it relates to other structural properties of hypergraphs.

\enlargethispage{\baselineskip}

%

{
\bibliography{jab}
\bibliographystyle{abbrv}
}
\end{document}